\documentclass[11pt]{article}
\usepackage[a4paper,margin=27mm]{geometry}
\usepackage{amsmath,amssymb,amsthm,mathtools}
\usepackage{fontspec}
\usepackage{etoolbox}
\usepackage{microtype,booktabs}
\usepackage{tikz,needspace}
\usetikzlibrary{arrows.meta,calc}
\pretocmd{\section}{\Needspace{8\baselineskip}}{}{}
\pretocmd{\subsection}{\Needspace{6\baselineskip}}{}{}
\usepackage[hidelinks,unicode]{hyperref}
\usepackage{bookmark}
\allowdisplaybreaks[2]
\DeclareMathOperator{\E}{\mathbb{E}}
\renewcommand{\P}{\mathop{\mathbb{P}}\nolimits}
\DeclareMathOperator{\SO}{SO}
\DeclareMathOperator{\BO}{BO}
\DeclareMathOperator{\RO}{RO}
\DeclareMathOperator{\FB}{FB}
\DeclareMathOperator*{\argmax}{arg\,max}
\DeclareMathOperator*{\argmin}{arg\,min}
\newtheorem{theorem}{Theorem}[section]
\newtheorem{lemma}[theorem]{Lemma}
\newtheorem{proposition}[theorem]{Proposition}
\newtheorem{observation}[theorem]{Observation}
\newtheorem*{mainresult}{Main result}
\theoremstyle{definition}
\newtheorem{definition}[theorem]{Definition}
\theoremstyle{plain}

\newenvironment{proofsketch}[1][Proof sketch]
{\begin{proof}[#1]}
    {\end{proof}}
\title{An Improved Upper Bound for the Random-Offerer Mechanism in Bilateral Trade via Recursive Hard Instances}
\author{%
  \begin{tabular}{cc}
    Heyang Shang & Jing Chen\\[0.3em]
    \normalsize Tsinghua University & \normalsize Tsinghua University\\[0.2em]
    \small\href{mailto:shanghy22@mails.tsinghua.edu.cn}{\texttt{shanghy22@mails.tsinghua.edu.cn}} &
    \small\href{mailto:jchencs@mail.tsinghua.edu.cn}{\texttt{jchencs@mail.tsinghua.edu.cn}}
  \end{tabular}%
}
\date{}
\makeatletter
\renewcommand{\@maketitle}{%
  \newpage\null\vskip 1em
  \begin{center}
    {\Large\bfseries\@title\par}
    \ifx\@author\@empty\else\vskip 1em{\large\@author\par}\fi
    \ifx\@date\@empty\else\vskip 0.5em{\@date\par}\fi
  \end{center}\vskip 1em}
\makeatother
\hypersetup{%
  pdftitle={An Improved Upper Bound for the Random-Offerer Mechanism in Bilateral Trade via Recursive Hard Instances},
  pdfauthor={Heyang Shang, Jing Chen}
}

\begin{document}
\maketitle
\begin{abstract}
  We study the worst-case efficiency of the random-offerer mechanism in bilateral trade relative to the first-best gains from trade. We construct a family of independent buyer value and seller cost distributions whose ratio converges to \(0.436943488\ldots\), improving the previous upper bound of \(0.460242308\ldots\). Our construction combines recursively nested seller distributions with downward transport of buyer values. The recursive structure allows contributions from different scales to accumulate, leading to stronger hard instances. We analyze this construction through a recurrence and obtain an analytic characterization of the limiting constant.
\end{abstract}

\section{Introduction}\label{sec:introduction}

Bilateral trade is a basic model for studying trade efficiency under incentive
constraints. A seller owns an item, and a buyer wishes to acquire it. Their
cost and value are drawn independently from known distributions, but each agent
privately observes their own realization. Trade creates positive gains from
trade whenever the buyer's value exceeds the seller's cost. Without incentive
constraints, trading in every such case achieves the first-best gains from
trade. The classical impossibility theorem of Myerson and
Satterthwaite~\cite{MS83}, however, shows that under standard nondegeneracy
conditions, full efficiency cannot be achieved together with Bayesian incentive
compatibility, individual rationality, and budget balance. This raises an
approximation question: how much of the first-best gains from trade can a
mechanism satisfying these constraints retain?

Characterizing the optimal feasible mechanism generally requires a detailed
analysis of both agents' distributions. Studying the efficiency guarantees of
simple mechanisms therefore provides one approach to this question. The
random-offerer mechanism, introduced by Br\"ustle, Cai, Wu, and
Zhao~\cite{BCWZ17}, chooses the seller or the buyer with equal probability to
make an offer maximizing their expected profit; the other agent then decides
whether to accept. This mechanism is known to achieve at least half of the
second-best gains from trade, the maximum attainable under Bayesian incentive
compatibility, individual rationality, and budget balance. We study its
performance relative to the first-best benchmark and denote its worst-case
ratio by
\[
  \rho_{\mathrm{RO}}
  :=\inf_{0<\FB<\infty}\frac{\RO}{\FB},
\]
where the infimum ranges over independent buyer value and seller cost
distributions. Determining this ratio requires both guarantees valid for all
distributions and hard instances that exhibit the mechanism's efficiency loss.

The \(1/2\) guarantee relative to the second-best benchmark does not directly
extend to the first-best benchmark. Cai, Goldner, Ma, and
Zhao~\cite{CGMZ21} constructed bilateral trade instances with
\(\RO/\FB<1/2\). Independently, Babaioff, Dobzinski, and
Kupfer~\cite{BDK21} gave an explicit instance with a ratio below \(0.495\). Cai,
Gupta, Li, and Mehta~\cite{CGLM26} used AI-guided evolutionary search to find
seller distributions and paired them with discrete equal-revenue buyer
distributions. Their instance satisfies \(\FB/\RO\ge2.0749\), corresponding to
an upper bound of approximately \(0.481951\) on \(\RO/\FB\). Jo~\cite{Jo26}
subsequently gave an explicit family of distributions that improved the upper
bound on \(\rho_{\mathrm{RO}}\) to \(0.460242308085529\). These results show
that even with one buyer and one seller, the random-offerer mechanism can lose
more than half of the first-best gains from trade.

In the other direction, Deng, Mao, Sivan, and Wang~\cite{DMSW22} established
the first constant lower bound,
\(\rho_{\mathrm{RO}}\ge1/8.23\approx0.121507\).
Fei~\cite{Fei22} subsequently improved the lower bound to approximately
\(0.317844\), Hartline and Wang~\cite{HW25} provided a corresponding geometric
analysis, and Jo~\cite{Jo26} improved the lower bound to \(1/\pi\). Bei, Wu,
and Zhou~\cite{BWZ26} further proved an \(e\)-approximation guarantee for a
generalized random-offerer mechanism in matching markets, yielding
\(\rho_{\mathrm{RO}}\ge1/e\) in bilateral trade. The worst-case ratio in
bilateral trade remains undetermined.

\subsection{Our result}

We construct a family of independent buyer and seller distributions that
further improves the upper bound on the worst-case ratio of the random-offerer
mechanism to \(0.436943488\ldots\). The exact constant \(\rho_*\) is defined by
the system of equations in Section~\ref{subsec:parameters}.

\begin{mainresult}
  For every \(\eta>0\), there exist independent buyer value and seller cost
  distributions such that \(0<\FB<\infty\) and
  \[
    \left|\frac{\RO}{\FB}-\rho_*\right|<\eta,
    \qquad \rho_*\approx0.436943488.
  \]
  Hence
  \[
    \rho_{\mathrm{RO}}\le\rho_*<0.436944.
  \]
\end{mainresult}

Together with the existing efficiency guarantee, this gives
\[
  \frac1e\le\rho_{\mathrm{RO}}\le\rho_*<0.436944.
\]
for the worst-case ratio of the random-offerer mechanism in bilateral trade.
The upper bound is approached by distributions obtained through finitely many
recursive steps; the infinite-layer distribution itself need not be used as a
trade instance. Under the \(\FB/\RO\) convention used by Cai et
al.~\cite{CGLM26}, our result corresponds to a worst-case approximation factor
of at least \(1/\rho_*\approx2.28862548\).

\subsection{Overview of the construction}

Our construction starts with a seller distribution that can be analyzed
exactly. It has an atom at its lowest cost, and its cumulative distribution
function is hyperbolic above that cost. The buyer distribution is obtained by
moving values downward from a truncated equal-revenue distribution. Preserving
the mass at the highest value and not increasing the tail probability at any
offer price preserves the seller's optimal offers and the gains from trade of
the seller-offering mechanism. Subject to this constraint, moving each buyer
value to an optimal point no higher than its original value minimizes its
contribution to \(\BO-\lambda\FB\) pointwise. The search for hard instances
thus reduces to analyzing the surplus-gap function for a fixed seller
distribution.

To further reduce the ratio, we scale an existing seller distribution down and
use it to replace the atom at the lowest cost, while retaining the outer
hyperbolic structure. Repeating this operation yields what we call
\emph{recursive seller distributions}: each step embeds a scaled copy of the
preceding distribution beneath a new hyperbolic layer. The equal-revenue
density used in the buyer construction has the form \(u^{-2}\), which exactly
offsets the scaling of costs and preserves the inner distribution's
contribution to the normalized surplus gap. The main recurrence decomposes the
contribution after adding one layer into an inner term scaled by the embedded
probability mass and an explicit outer term, with an error bound independent of
the inner distribution and the recursion depth. Unrolling the recurrence
reduces the recursive construction to a geometric series, and optimizing two
probability parameters determines the constant above.

Section~\ref{sec:model} introduces the model and expresses gains from trade as
conditional functions for a fixed seller distribution. Section~\ref{sec:ideal}
analyzes the simple single-layer distribution and the downward movement of
buyer values. Section~\ref{sec:construction} presents the recursive hard
instances, the recurrence, and the parameter choice, and proves the main
result. Detailed proofs of the recurrence and the optimal-offer results appear
in Appendix~\ref{sec:proofs}.

\section{Model}\label{sec:model}

\subsection{Bilateral trade}

Consider a buyer and a seller who owns a single item.
The seller's cost is a random variable \(S\ge0\), and the buyer's value for
the item is a random variable \(B\ge0\).
They are independent, with cumulative distribution functions \(F\) and
\(F_B\), respectively. Both distributions are known to both agents, but each
agent knows only their own type.
We allow atoms and consistently use \(B\ge p\) and \(S\le p\) as the conditions
for accepting an offer at price \(p\).

Both agents have quasilinear utilities. If trade occurs at price \(p\), the
buyer receives utility \(B-p\), and the seller receives utility \(p-S\).
Otherwise, both utilities are \(0\).
The gains from a trade are therefore \(B-S\), independently of its price.
Throughout, we measure efficiency by expected gains from trade.

\subsection{Efficiency benchmarks and the random-offerer mechanism}

\subsubsection{First-best gains from trade}

The first-best mechanism ignores incentive constraints and trades if and
only if \(B\ge S\). Its expected gains from trade are
\begin{equation*}
  \FB(F,F_B):=\E[(B-S)_+],
  \qquad (t)_+:=\max\{t,0\}.
\end{equation*}
We consider pairs of distributions satisfying \(0<\FB<\infty\).

\subsubsection{Seller-offering mechanism}

A seller with cost \(s\) chooses an offer maximizing their expected profit
under the buyer distribution:
\begin{equation}
  p_{\mathrm{S}}(s)\in
  \argmax_{p\ge s}(p-s)\P(B\ge p).
  \label{eq:sellerprice}
\end{equation}
The buyer accepts if \(B\ge p_{\mathrm{S}}(s)\) and pays this price to the seller.
When several offers are optimal and the maximum profit is positive, the
seller chooses the highest such price. When the maximum profit is zero,
we prescribe the offer \(s\).
The expected gains from trade of the seller-offering mechanism are denoted by
\begin{equation*}
  \SO(F,F_B):=
  \E\bigl[(B-S)\mathbf{1}_{\{B\ge p_{\mathrm{S}}(S)\}}\bigr].
\end{equation*}

\subsubsection{Buyer-offering mechanism}

A buyer with value \(b\) chooses an offer maximizing their expected profit
under the seller distribution.
When several offers are optimal, the buyer chooses the lowest such price:
\begin{equation}
  p_F(b):=\min\argmax_{0\le p\le b}(b-p)F(p).
  \label{eq:buyerprice}
\end{equation}
The seller accepts if \(S\le p_F(b)\), and trade occurs at price \(p_F(b)\).
The expected gains from trade of the buyer-offering mechanism are denoted by
\begin{equation*}
  \BO(F,F_B):=
  \E\bigl[(B-S)\mathbf{1}_{\{S\le p_F(B)\}}\bigr].
\end{equation*}

\subsubsection{Random-offerer mechanism}

The random-offerer mechanism runs the seller-offering and buyer-offering
mechanisms with equal probability. Its expected gains from trade are
\begin{equation}
  \RO(F,F_B):=\frac{\SO(F,F_B)+\BO(F,F_B)}2.
  \label{eq:RO}
\end{equation}
When the pair of distributions is fixed, we omit the distribution arguments
and write \(\FB,\SO,\BO,\RO\).
Equations~\eqref{eq:sellerprice} and~\eqref{eq:buyerprice} maximize the offering
agent's expected profit, whereas \(\SO\) and \(\BO\) measure the sum of the
agents' utilities, namely the gains from trade.

\subsection{Worst-case ratio and the construction objective}

We study the worst-case ratio of the random-offerer mechanism's gains from
trade to the first-best gains from trade:
\begin{equation*}
  \rho_{\mathrm{RO}}
  :=\inf_{\substack{F,F_B\\0<\FB(F,F_B)<\infty}}
  \frac{\RO(F,F_B)}{\FB(F,F_B)}.
\end{equation*}
Thus, any pair of distributions with a small ratio gives an upper bound on
\(\rho_{\mathrm{RO}}\).

Fix \(0<\lambda<1\). By~\eqref{eq:RO}, every instance satisfying
\(0<\FB<\infty\) has
\begin{equation*}
  \frac{\RO}{\FB}-\frac\lambda2
  =\frac{\SO+\BO-\lambda\FB}{2\FB}.
\end{equation*}
Hence, an instance with \(\SO+\BO-\lambda\FB\le0\) gives
\(\rho_{\mathrm{RO}}\le\lambda/2\).
For a family of approximating instances, controlling the surplus gap relative
to \(\FB\) can likewise make the ratio converge to \(\lambda/2\).
Our construction focuses on this surplus gap.

\subsection{Conditional gains from trade for a fixed seller distribution}
\label{subsec:conditionalgft}

For a fixed seller distribution \(F\), the buyer's offer depends only on the
buyer value \(b\). We first express the first-best and buyer-offering gains
from trade as functions of \(b\), and then take expectations over the buyer
distribution.
The first-best gains from trade at buyer value \(b\) are
\begin{equation}
  D_F(b):=\E[(b-S)_+]=\int_0^b F(s)\,ds,
  \label{eq:D}
\end{equation}
The gains from trade of the buyer-offering mechanism are
\begin{equation}
  C_F(b):=\E[(b-S)\mathbf{1}_{\{S\le p_F(b)\}}]
  =D_F(p_F(b))+(b-p_F(b))F(p_F(b)).
  \label{eq:C}
\end{equation}
In~\eqref{eq:C}, the second term is the buyer's expected profit from the offer,
and the first term is the seller's expected utility from trade.
By independence,
\begin{equation*}
  \FB=\E[D_F(B)],\qquad \BO=\E[C_F(B)].
\end{equation*}

Define the surplus-gap function
\begin{equation*}
  h_{\lambda,F}(b):=C_F(b)-\lambda D_F(b).
\end{equation*}
It records the contribution of buyer value \(b\) to \(\BO-\lambda\FB\) and satisfies
\begin{equation}
  \SO+\BO-\lambda\FB
  =\SO+\E[h_{\lambda,F}(B)].
  \label{eq:gap}
\end{equation}
For a fixed seller distribution, the construction objective is therefore to
control the seller-offering gains from trade \(\SO\) while choosing a buyer
distribution that makes \(\E[h_{\lambda,F}(B)]\) as small as possible.
The next section uses a simple seller distribution to show how to preserve
\(\SO\) while minimizing this expectation pointwise.

\section{A simple distribution}\label{sec:ideal}
Before constructing the final hard instances, we analyze a single-layer seller
distribution that admits exact calculations and explain how to choose a
corresponding buyer distribution to minimize the surplus gap in~\eqref{eq:gap}.

Let \(0<x<1\) be a parameter, and define the seller's cumulative distribution
function by
\begin{equation*}
  Q_x(s)=
  \begin{cases}
    0,              & s<0,        \\
    \dfrac{x}{1-s}, & 0\le s<1-x, \\
    1,              & s\ge1-x.
  \end{cases}
\end{equation*}
In this section, when there is no ambiguity, we abbreviate the distribution
subscript \(Q_x\) on \(C,D,h,H\) as \(x\).
When several offers are optimal, the seller chooses the highest price,
and the buyer chooses the lowest price.\footnote{These conventions can be
  implemented in the limit by arbitrarily small perturbations of our hard
  instances, as in Babaioff, Dobzinski, and Kupfer~\cite{BDK21}.
  The perturbations eliminate ties that affect gains from trade while changing
  those gains by an arbitrarily small amount, so the upper bound holds under
  any tie-breaking rule.}

\subsection{Buyer offers and \texorpdfstring{$C,D$}{C and D}}
For \(0\le p<1-x\),
\[
  \frac{d}{dp}\left(\frac{x(b-p)}{1-p}\right)
  =\frac{x(b-1)}{(1-p)^2}.
\]
Thus, when \(b<1\), the buyer's optimal offer is \(0\). When \(b=1\), all
prices in \([0,1-x]\) yield the same profit, so our tie-breaking convention
again selects \(0\). When \(b>1\), the buyer's optimal offer is \(1-x\).
Equations~\eqref{eq:D}--\eqref{eq:C} give
\begin{equation*}
  \begin{aligned}
    C_x(b) & =
    \begin{cases}
      xb,     & 0\le b\le1, \\
      D_x(b), & b>1,
    \end{cases} \\
    D_x(b) & =
    \begin{cases}
      -x\log(1-b),   & 0\le b\le1-x, \\
      b-1+x-x\log x, & b\ge1-x.
    \end{cases}
  \end{aligned}
\end{equation*}
Hence, on \([0,1]\),
\begin{equation}
  h_{\lambda,x}(b)=
  \begin{cases}
    x\bigl[b+\lambda\log(1-b)\bigr],     & 0\le b\le1-x, \\
    xb-\lambda\bigl[b-1+x-x\log x\bigr], & 1-x\le b\le1.
  \end{cases}
  \label{eq:idealh}
\end{equation}
For \(b>1\), we have \(h_{\lambda,x}(b)=(1-\lambda)D_x(b)>0\).

With the seller distribution \(Q_x\) fixed, \(h_{\lambda,x}(b)\) now explicitly
describes the contribution of each buyer value \(b\) to
\(\mathrm{BO}-\lambda\mathrm{FB}\). We next choose a buyer distribution that
makes these contributions as small as possible, while also controlling the
seller-offering gains from trade \(\mathrm{SO}\). To meet both objectives, we
start with a baseline buyer distribution under which the seller's optimal
offers have a simple structure, then adjust its probability mass while
preserving this offer property.

\subsection{Seller offers and downward transport of buyer values}

Define a baseline buyer value \(\widetilde B\) with tail distribution
\begin{equation*}
  \overline F_{\widetilde B}(u):=\P(\widetilde B\ge u)=
  \begin{cases}
    1,         & u\le a,  \\
    \dfrac au, & a<u\le1, \\
    0,         & u>1.
  \end{cases}
\end{equation*}
where \(0<a<1\) is a parameter.
This distribution is supported on \([a,1]\) and has an atom of mass \(a\) at \(1\).

\begin{observation}
  Under the baseline buyer distribution \(\widetilde B\), every seller with
  cost \(s\in[0,1-x]\) offers price \(1\).
\end{observation}

\begin{lemma}\label{lem:preserveSellerPrice}
  Suppose the buyer value \(B\) satisfies
  \[
    B\in[0,1],\qquad
    \P(B=1)=a,\qquad
    \P(B\ge p)\le \frac ap\quad(0<p<1),
  \]
  Then every seller with cost \(s\in[0,1-x]\) offers price \(1\), and the
  expected gains from trade of the seller-offering mechanism are
  \[
    \SO=aD_x(1).
  \]
\end{lemma}

\begin{proof}
  Fix \(s\in[0,1-x]\). Offering price \(1\) yields profit
  \((1-s)\P(B\ge1)=a(1-s)>0\).
  For any \(0<p<1\) with \(p\ge s\),
  \[
    (p-s)\P(B\ge p)
    \le (p-s)\frac ap
    =a\left(1-\frac sp\right)
    \le a(1-s).
  \]
  The last inequality is strict when \(s>0\). Offers \(p\le s\) yield nonpositive
  profit, while offers \(p>1\) are never accepted and yield zero profit. Thus,
  when \(s>0\), the unique optimal offer is \(1\). When \(s=0\), price \(1\) is
  still optimal and is the highest optimal price, so our convention also selects
  \(1\).

  Trade therefore occurs exactly when \(B=1\), giving
  \[
    \SO=\E\bigl[(1-S)\mathbf{1}_{\{B=1\}}\bigr]
    =\P(B=1)\E[1-S]
    =aD_x(1).\qedhere
  \]
\end{proof}

Lemma~\ref{lem:preserveSellerPrice} shows that preserving the mass at $1$
without increasing the buyer's tail probabilities leaves $\SO$ unchanged. We
therefore consider \emph{downward transport} of the baseline distribution
$\widetilde B$: take a map
\[
  \phi:[0,1]\to[0,1],\qquad \phi(u)\le u,\qquad \phi(1)=1,
\]
and let $B=\phi(\widetilde B)$. That is, mass originally at $u<1$ may only be
moved to a value no greater than $u$, while the atom at $1$ is preserved.

Since $\phi(u)\le u$, for every $0<p<1$, $\P(B\ge p)\le \P(\widetilde B\ge
  p)=\min\{1,a/p\}\le a/p$, and $\P(B=1)=\P(\widetilde B=1)=a$.
Lemma~\ref{lem:preserveSellerPrice} thus still applies, and every such downward
transport preserves $\SO=aD_x(1)$.

Optimizing the random-offerer mechanism over these transports therefore
requires only minimizing $\BO-\lambda\FB=\E[h_{\lambda,x}(B)]$. We next choose
the optimal destination pointwise.

\subsection{Optimal downward transport}
For a baseline buyer value \(u<1\), the permitted destinations form \([0,u]\).
To minimize its contribution to \(\BO-\lambda\FB\), we should therefore choose
a minimizer of \(h_{\lambda,x}\) on this interval.

\begin{definition}
  For a seller distribution \(F\), whenever the following minima are attained,
  define
  \[
    H_{\lambda,F}(u)
    =
    \min_{0\le b\le u}h_{\lambda,F}(b),
    \qquad
    m_{\lambda,F}(u)
    =
    \min\argmin_{0\le b\le u}h_{\lambda,F}(b).
  \]
\end{definition}

\begin{definition}
  The buyer value \(B\) is obtained by transporting \(\widetilde B\) downward:
  each value strictly below \(1\) is moved to the lowest minimizer of
  \(h_{\lambda,x}\) at or below that value, while the mass at \(1\) is preserved.
  Specifically, define
  \[
    B=
    \begin{cases}
      m_{\lambda,x}(\widetilde B),
       & \widetilde B<1, \\[1mm]
      1,
       & \widetilde B=1.
    \end{cases}
  \]

  The construction immediately gives
  \[
    B\le \widetilde B\le1,
    \qquad
    \P(B=1)=\P(\widetilde B=1)=a.
  \]
\end{definition}

\subsection{Characterizing \texorpdfstring{$m_{\lambda,x}$}{m}}
Recall that on \([0,1]\),
\begin{equation*}
  h_{\lambda,x}(b)=
  \begin{cases}
    x\bigl[b+\lambda\log(1-b)\bigr],     & 0\le b\le1-x, \\
    xb-\lambda\bigl[b-1+x-x\log x\bigr], & 1-x\le b\le1.
  \end{cases}
\end{equation*}
For \(b>1\), we have \(h_{\lambda,x}(b)=(1-\lambda)D_x(b)>0\).

The function \(b+\lambda\log(1-b)\) has a unique zero \(\tau=\tau_\lambda\) in
\((0,1)\), satisfying
\begin{equation}
  \tau+\lambda\log(1-\tau)=0,\qquad 1-\lambda<\tau<1.
  \label{eq:tau}
\end{equation}
From now on, take $0<x<1-\tau$.
Then \(x<\lambda\).
The function \(h_{\lambda,x}\) is nonnegative on \([0,\tau]\) and strictly
decreasing on \([\tau,1]\), so
\begin{equation}
  H_{\lambda,x}(u)=
  \begin{cases}
    0,                & 0\le u\le\tau, \\
    h_{\lambda,x}(u), & \tau\le u\le1, \\
    h_{\lambda,x}(1), & u\ge1.
  \end{cases}
  \label{eq:idealH}
\end{equation}
Equivalently,
\begin{equation*}
  \begin{aligned}
    m_{\lambda,x}(u)=
    \begin{cases}
      0, & 0\le u\le\tau, \\
      u, & \tau<u<1,      \\
      1, & u\ge1.
    \end{cases}
  \end{aligned}
\end{equation*}

The optimal transport thus moves buyer values at most \(\tau\) to \(0\) and
leaves higher values unchanged. In particular, \(m_{\lambda,x}(1)=1\), so the
requirement that the endpoint atom be preserved agrees with pointwise
minimization. The entire buyer construction can therefore be written as
\(B=m_{\lambda,x}(\widetilde B)\).

Equivalently, take an auxiliary random variable \(U\) with density
\[
  f_U(u)=\frac{a}{u^2}\mathbf{1}_{\{u\ge a\}},
\]
Then \(\widetilde B=\min\{U,1\}\). Since \(m_{\lambda,x}(u)=1\) for all
\(u\ge1\), the same buyer construction can also be written as
\[
  B=m_{\lambda,x}(U).
\]

\subsection{Calculations for the simple distribution}
Take \(0<a\le\tau\). Since \(H_{\lambda,x}(u)=0\) for \(u\le\tau\) and
\(h_{\lambda,x}(B)=H_{\lambda,x}(U)\), the surplus gap \(\BO-\lambda\FB\) for
this instance satisfies
\begin{equation*}
  \dfrac{\BO -\lambda \FB}{a} = \frac{\E[h_{\lambda,x}(B)]}{a}= \int_0^\infty\frac{H_{\lambda,x}(u)}{u^2}\,du:=J_\lambda(x).
\end{equation*}
Direct calculation gives
\begin{equation}
  \begin{aligned}
    J_\lambda(x) & =xA_\lambda+
    \lambda\bigl[x\log x+(1-x)\log(1-x)\bigr].
  \end{aligned}
  \label{eq:Jclosed}
\end{equation}
where \(A_\lambda=\tau+(\lambda-1)\log\tau\).
Consequently,
\begin{equation*}
  \begin{aligned}
    \dfrac{\SO+\BO -\lambda \FB}{a} = D_x(1)+J_\lambda(x).
  \end{aligned}
\end{equation*}
Optimizing over $x$ gives the best $\lambda$ attainable by this single-layer
construction. In fact, this construction cannot yield a ratio $\RO/\FB$
strictly below $1/2$.

Although the single-layer construction is now fully analyzed, its constant is
not small enough. In the next section, we therefore seek a recursive
construction that accumulates the surplus gaps of individual layers.

\section{Recursive hard instances and the final constant}
\label{sec:construction}
Building on the simple construction in the previous section, we now construct
recursive hard instances.
The idea is to replace the atom of mass \(x\) at \(0\) in \(Q_x\) with an existing
seller distribution scaled down to the cost interval \([0,\varepsilon]\),
while retaining an outer hyperbolic structure similar to that of \(Q_x\).

This operation can be repeated recursively, embedding similar structures at
different scales.
Although each step scales the previous distribution down by a factor of
\(\varepsilon\), its contribution does not vanish: the corresponding buyer
construction has a density of the form \(u^{-2}\).
As \(\varepsilon\to0\), the previous layer's contribution is approximately
multiplied only by the embedded probability mass, and the new outer layer
contributes a single-layer term \(J_\lambda(x)\) depending only on its own
parameter.

We first define the recursive transformation and the corresponding buyer
distribution.
We then prove the main recurrence: adding an outer layer with parameter \(x\)
approximately preserves \(x\) times the previous surplus gap and adds the
single-layer contribution \(J_\lambda(x)\).
Repeated application reduces the analysis of the recursive construction to a geometric series
and, ultimately, to a finite-dimensional parameter optimization.

\subsection{Seller and buyer distributions}
\subsubsection{The recursive seller distribution family}
Suppose the seller distribution \(F\) is supported on \([0,z]\), where
\(z\le1\) is its upper support endpoint.
For \(x,\varepsilon\in(0,1)\), let \(\ell=\varepsilon z\),
\(z'=1-x(1-\ell)\), and define
\begin{equation*}
  (T_{x,\varepsilon}F)(s)=
  \begin{cases}
    xF(s/\varepsilon),      & s\le\ell,  \\
    \dfrac{x(1-\ell)}{1-s}, & \ell<s<z', \\
    1,                      & s\ge z'.
  \end{cases}
\end{equation*}

\begin{definition}[Recursive seller distribution family]\label{def:recursiveFamily}
  Let \(\mathcal F\) be the family of seller distributions obtained from
  \(\delta_1\) by finitely many transformations \(T_{x,\varepsilon}\):
  \[
    \mathcal F=
    \left\{T_{x_n,\varepsilon_n}\cdots T_{x_1,\varepsilon_1}\delta_1:
    n\ge0,\quad x_i,\varepsilon_i\in(0,1)\right\}.
  \]
  When \(n=0\), the expression denotes \(\delta_1\).
  Parameters may differ across layers, so \(T_{x,\varepsilon}F\in\mathcal F\)
  for every \(F\in\mathcal F\).
\end{definition}

We now specify the recursive construction used below.
Fix \(0<\lambda<1\), and let \(\tau=\tau_\lambda\) satisfy~\eqref{eq:tau}.
For \(0<q,r<1-\tau\), \(N\in\mathbb Z_{\ge0}\), and \(0<\varepsilon<1\), define
\begin{equation}
  \begin{aligned}
    X_{q,\varepsilon,0}   & =\delta_1,                              \\
    X_{q,\varepsilon,k}   & =T_{q,\varepsilon}X_{q,\varepsilon,k-1}
    \quad(1\le k\le N),                                             \\
    G_{q,r,N,\varepsilon} & =T_{r,\varepsilon}X_{q,\varepsilon,N}.
  \end{aligned}
  \label{eq:family}
\end{equation}

The distribution \(X_{q,\varepsilon,N}\) is the inner distribution obtained
by repeatedly using parameter \(q\).
The distribution \(G_{q,r,N,\varepsilon}\) adds an outer layer with parameter
\(r\) and is the seller distribution used in the hard instance.
Each transformation scales the lowest support point by \(\varepsilon\), so
the final instance has lowest support point
\[
  a(G_{q,r,N,\varepsilon})=\varepsilon^{N+1}>0.
\]

\subsubsection{The corresponding buyer distribution}
Having defined the recursive seller distribution family, we again seek to
minimize the surplus gap from buyer offers while controlling seller offers.
We construct the buyer distribution in the same way as in the previous
section.
For each \(F\in\mathcal F\), let \(a(F)>0\) be its lowest support point.
Independently of the seller, draw a random variable \(U_F\) with density
\begin{equation}
  f_{U_F}(u)=\frac{a(F)}{u^2}\mathbf{1}_{\{u\ge a(F)\}}.
  \label{eq:U}
\end{equation}
Recall that \(H_{\lambda,F}(u)=\min_{0\le b\le u}h_{\lambda,F}(b)\).
Define the buyer value corresponding to \(F\) by
\begin{equation}
  B_{\lambda,F}=m_{\lambda,F}(U_F),\qquad
  m_{\lambda,F}(u)=\min\argmin_{0\le b\le u}h_{\lambda,F}(b).
  \label{eq:buyer}
\end{equation}

When analyzing a fixed seller distribution, abbreviate \(a=a(F)\),
\(U=U_F\), and \(B=B_{\lambda,F}\), and omit the fixed subscript \(\lambda\)
from \(h_F,H_F\).

Since \(B\le U\), for every \(p>0\),
\begin{equation*}
  \P(B\ge p)\le\P(U\ge p)
  =\min\left\{1,\frac ap\right\}\le\frac ap.
\end{equation*}
Also, \(h_{\lambda,F}(B)=H_{\lambda,F}(U)\), so define
\begin{equation}
  L_\lambda(F):=\int_{a(F)}^\infty\frac{H_{\lambda,F}(u)}{u^2}\,du
  =\frac{\BO-\lambda\FB}{a(F)}.
  \label{eq:L}
\end{equation}
Here \(\BO,\FB\) refer to the seller distribution \(F\) paired with buyer
value \(B_{\lambda,F}\).
Appendix~\ref{subsec:formulas} will also establish
\begin{equation}
  H_{\lambda,F}(u)=0\quad(0\le u<a(F)),
  \label{eq:Htail}
\end{equation}
Thus, the lower limit in~\eqref{eq:L} can equivalently be written as \(0\).

\subsection{The recurrence and gains from trade}

We now study the recurrence described at the beginning of this section.
One transformation \(T_{x,\varepsilon}\) scales an existing distribution
down and embeds it in the low-cost interval with probability mass \(x\),
while adding a new hyperbolic layer outside it.
We show below that these two parts contribute approximately
\(xL_\lambda(F)\) and \(J_\lambda(x)\), respectively.

\begin{lemma}\label{prop:recursion}
  Fix \(0<\lambda<1\), and let \(\tau\) satisfy~\eqref{eq:tau}.
  If \(0<x<1-\tau\), then there exist constants \(C>0\) and
  \(\varepsilon_0>0\), depending only on \(\lambda,x\), such that for every
  \(F\in\mathcal F\) and \(0<\varepsilon<\varepsilon_0\),
  \begin{equation}
    \left|
    L_\lambda(T_{x,\varepsilon}F)
    -xL_\lambda(F)-J_\lambda(x)
    \right|
    \le C\varepsilon.
    \label{eq:recursion}
  \end{equation}
\end{lemma}

\begin{proofsketch}
  Write \(G=T_{x,\varepsilon}F\), and temporarily omit the fixed subscript
  \(\lambda\) from \(H\).

  On \(0\le b\le\varepsilon\), the inner seller distribution \(F\) has its
  costs scaled by \(\varepsilon\) and retains total probability mass \(x\).
  Thus, at buyer value \(b\),
  \[
    h_G(b)
    =
    x\varepsilon h_F(b/\varepsilon)+O(b^2),
    \qquad 0\le b\le\varepsilon,
  \]

  Since $H_G(u)=\min_{0\le b\le u}h_G(b)$, taking the minimum over
  \(b\in[0,u]\) gives
  \[
    H_G(u)
    =
    x\varepsilon H_F(u/\varepsilon)+O(u^2),
    \qquad 0\le u\le\varepsilon.
  \]
  In particular, setting \(u=\varepsilon t\) gives, for \(0\le t\le1\),
  \[
    H_G(\varepsilon t)
    =
    x\varepsilon H_F(t)+O(\varepsilon^2t^2).
  \]

  Consider the contribution of \([0,\varepsilon]\) to \(L_\lambda(G)\).
  With \(u=\varepsilon t\),
  \[
    \begin{aligned}
      \int_0^\varepsilon
      \frac{H_G(u)}{u^2}\,du
       & =
      \int_0^1
      \frac{x\varepsilon H_F(t)+O(\varepsilon^2t^2)}
      {\varepsilon^2t^2}\,
      \varepsilon\,dt \\
       & =
      x\int_0^1\frac{H_F(t)}{t^2}\,dt
      +O(\varepsilon).
    \end{aligned}
  \]
  This calculation shows why the contribution survives the scaling.
  Scaling contributes a factor of \(\varepsilon\) to \(H_G\), while the
  change of variables contributes \(\varepsilon^{-2}\) from \(u^{-2}\) and
  another \(\varepsilon\) from the integration element.
  These three factors cancel, leaving only the probability mass \(x\).
  The inner construction's contribution therefore remains of constant order
  even at arbitrarily small scales, rather than vanishing with \(\varepsilon\).

  More precisely, the calculation so far gives only
  \[
    x\int_0^1\frac{H_F(t)}{t^2}\,dt,
  \]
  The missing part, $x\int_1^\infty\frac{H_F(t)}{t^2}\,dt$, is exactly
  recovered on the next interval
  \([\varepsilon,\tau]\), while the contribution from \([\tau,\infty)\)
  approaches the simple distribution's contribution \(J_\lambda(x)\).
  Appendix~\ref{subsec:lemmaproofs} gives uniform error estimates on these
  three intervals. Combining them yields
  \[
    L_\lambda(T_{x,\varepsilon}F)
    =
    xL_\lambda(F)+J_\lambda(x)+O(\varepsilon).
  \]
  The two terms in the recurrence thus have direct interpretations:
  \(xL_\lambda(F)\) is the contribution retained from the old distribution
  after scaling and embedding it at the lower end, and \(J_\lambda(x)\)
  is the contribution of the newly added outer hyperbolic layer.
\end{proofsketch}

Applying the preceding lemma recursively gives the following property of
the seller distribution $G_{q,r,N,\varepsilon}$ used in our construction.
\begin{lemma}\label{lem:unroll}
  Fix \(0<\lambda<1\), let \(\tau\) satisfy~\eqref{eq:tau}, and take
  \(0<q,r<1-\tau\).
  There exists \(\varepsilon_0>0\), depending only on \(\lambda,q,r\),
  such that for every \(N\in\mathbb Z_{\ge0}\) and
  \(0<\varepsilon<\varepsilon_0\),
  \begin{equation}
    L_\lambda(G_{q,r,N,\varepsilon})
    =r\frac{1-q^N}{1-q}J_\lambda(q)+J_\lambda(r)+O(\varepsilon).
    \label{eq:unroll}
  \end{equation}
  The error constant depends only on \(\lambda,q,r\), independently of
  \(N\) and \(\varepsilon\).
\end{lemma}

\begin{proofsketch}
  Apply~\eqref{eq:recursion} repeatedly.
  The errors at successive layers are weighted by powers of \(q\), and their
  sum is bounded by a convergent geometric series.
  The total error is therefore \(O(\varepsilon)\), with an implicit constant
  independent of the number of iterations \(N\).
\end{proofsketch}

The recurrence above describes the contribution
\(\mathrm{BO}-\lambda\mathrm{FB}\) in the recursive construction.
To obtain the full surplus gap
\(\mathrm{SO}+\mathrm{BO}-\lambda\mathrm{FB}\), we still need to analyze
the seller-offering gains from trade.
We show next that, under the final distribution
\(G_{q,r,N,\varepsilon}\), sellers continue to offer price \(1\).

\begin{lemma}\label{lem:bo}
  For \(G=G_{q,r,N,\varepsilon}\in\mathcal F\), take the corresponding buyer
  value \(B_{\lambda,G}\), and let \(a=a(G)\).
  There exists \(\varepsilon_0>0\), depending only on \(\lambda,r\), such
  that for all \(0<\varepsilon<\varepsilon_0\), every seller uniquely
  optimally offers price \(1\), and
  \begin{equation}
    \frac{\SO}{a}=D_G(1)=r(1-\log r)+O(\varepsilon).
    \label{eq:SOannounced}
  \end{equation}
\end{lemma}

The proofs of Lemmas~\ref{prop:recursion}, \ref{lem:unroll}, and~\ref{lem:bo}
appear in Appendices~\ref{subsec:lemmaproofs}, \ref{subsec:unrollproof},
and~\ref{subsec:SOproof}, respectively.
Combining these lemmas gives
\begin{equation*}
  \frac{\SO+\BO-\lambda\FB}{a}
    =r(1-\log r)+J_\lambda(r)+r\frac{1-q^N}{1-q}J_\lambda(q)+O(\varepsilon).
\end{equation*}
As \(N\to\infty\), denote the limit of the main term on the right-hand side by
\begin{equation}
  \Phi_\lambda(q,r)=r(1-\log r)+J_\lambda(r)+\frac{r}{1-q}J_\lambda(q).
  \label{eq:Phi}
\end{equation}

\subsection{Parameter choice and the main theorem}\label{subsec:parameters}

Within the construction in~\eqref{eq:Phi}, we seek the smallest \(\lambda\)
for which \(\Phi_\lambda(q,r)\le0\).
Since \(r>0\),
\begin{equation}
  \frac{\Phi_\lambda(q,r)}r
  =\frac{J_\lambda(q)}{1-q}+1-\log r+\frac{J_\lambda(r)}r.
  \label{eq:separate}
\end{equation}
Equation~\eqref{eq:separate} allows the two parameters to be optimized
separately. By~\eqref{eq:Jclosed},
\begin{equation*}
  \begin{aligned}
    \frac{d}{dq}\frac{J_\lambda(q)}{1-q}
     & =\frac{A_\lambda+\lambda\log q}{(1-q)^2}, \\
    \frac{d}{dr}\left(1-\log r+\frac{J_\lambda(r)}r\right)
     & =-\frac{r+\lambda\log(1-r)}{r^2}.
  \end{aligned}
\end{equation*}
Viewing the closed-form expressions as functions on \((0,1)\), the unique
minimizer of the first term satisfies
\begin{equation}
  A_\lambda+\lambda\log q=0,
  \label{eq:qminimum}
\end{equation}
and the unique minimizer of the second term is
\begin{equation}
  r=\tau_\lambda.
  \label{eq:rminimum}
\end{equation}
At these two points, respectively,
\[
  \frac{J_\lambda(q)}{1-q}=\lambda\log(1-q),\qquad
  1-\log\tau+\frac{J_\lambda(\tau)}\tau=2A_\lambda.
\]
Setting the minimum of~\eqref{eq:separate} equal to zero gives
\begin{equation}
  0=2A_\lambda+\lambda\log(1-q)
  =\lambda\log\frac{1-q}{q^2}.
  \label{eq:criticalzero}
\end{equation}
Thus, \(q^2+q=1\).
Combining~\eqref{eq:qminimum}, \eqref{eq:rminimum}, and~\eqref{eq:criticalzero}
with the defining equation~\eqref{eq:tau} for \(\tau\), the critical
parameters are defined by the system
\begin{equation}
  \begin{gathered}
      q_*=\frac{\sqrt5-1}{2},\qquad r_*=\tau_*,\\
      \tau_*+\lambda_*\log(1-\tau_*)=0,\\
      \tau_*+(\lambda_*-1)\log\tau_*+\lambda_*\log q_*=0.
    \end{gathered}
  \label{eq:parameters}
\end{equation}
Table~\ref{tab:parameters} lists the parameter values satisfying these conditions.
\begin{table}[htbp]
  \centering
  \small
  \begin{tabular}{@{}c c@{}}
    \toprule
    Parameter & Value \\
    \midrule
    \(q_*\)       & \(0.618033988\ldots\) \\
    \(r_*\)       & \(0.241149237\ldots\) \\
    \(\lambda_*\) & \(0.873886976\ldots\) \\
    \bottomrule
  \end{tabular}
  \caption{Optimal parameters}
  \label{tab:parameters}
\end{table}

Since \(\tau_*<1/4\) and \(q_*<3/4\), we have \(q_*<1-\tau_*\) and
\(r_*=\tau_*<1-\tau_*\).
These points therefore satisfy the original constraints, so the preceding
analysis applies to the distributions constructed with these parameters, and
\begin{equation}
  \Phi_{\lambda_*}(q_*,r_*)=0.
  \label{eq:Phizero}
\end{equation}

Write
\begin{equation*}
  \rho_*:=\frac{\lambda_*}{2}\approx0.436943488.
\end{equation*}

With the construction parameters determined, we can now prove the main result.

\begin{theorem}\label{thm:main}
  Take \(q_*,r_*,\lambda_*\) as defined by~\eqref{eq:parameters}.
  Let the seller distribution be \(G=G_{q_*,r_*,N,\varepsilon}\in\mathcal F\)
  from~\eqref{eq:family}, and let the independent buyer value
  \(B_{\lambda_*,G}\) be defined by~\eqref{eq:U}--\eqref{eq:buyer}, with
  \(a=\varepsilon^{N+1}\) and \(\lambda=\lambda_*\).
  For sufficiently small positive \(\varepsilon\), we have \(\FB>0\) and
  \begin{equation*}
    \frac{\RO}{\FB}=\rho_*+O(\varepsilon+q_*^N),
  \end{equation*}
  The error constant is independent of \(N\) and \(\varepsilon\).
  Thus, for every \(\eta>0\), there exist a finite \(N\) and a positive
  \(\varepsilon\) such that
  \[
    \left|\frac{\RO}{\FB}-\rho_*\right|<\eta.
  \]
  In particular, \(\rho_{\mathrm{RO}}\le\rho_*<0.436944\).
\end{theorem}

\begin{proof}
  Fix
  \[
    G=G_{q_*,r_*,N,\varepsilon},\qquad
    B=B_{\lambda_*,G},\qquad
    a=a(G)=\varepsilon^{N+1}.
  \]
  Throughout this proof, $\SO,\BO,\FB$ refer to this final instance.

  First, we give a positive lower bound on $\FB/a$ independent of
  $N,\varepsilon$.
  By~\eqref{eq:topatom}, $B\le 1$ and $\P(B=1)=a$.
  Keeping only the contribution of $\{B=1\}$ to the first-best gains from
  trade gives
  \[
    \FB=\E[(B-S)_+]
    \ge \P(B=1)\E[1-S]
    =aD_G(1).
  \]
  By~\eqref{eq:Dbound}, $D_G(1)=r_*(1-\log r_*)+O(\varepsilon)$.
  Since $r_*(1-\log r_*)>0$, for sufficiently small $\varepsilon$ there
  exists a constant $c>0$, depending only on $r_*$, such that
  \[
    \frac{\FB}{a}\ge D_G(1)\ge c.
  \]
  In particular, $\FB>0$.

  Next, we estimate $\SO+\BO-\lambda_*\FB$.
  Equations~\eqref{eq:unroll} and~\eqref{eq:SOexact} give
  \[
    \frac{\SO+\BO-\lambda_*\FB}{a}
    =
    r_*(1-\log r_*)+J_{\lambda_*}(r_*)
    +r_*\frac{1-(q_*)^N}{1-q_*}J_{\lambda_*}(q_*)
    +O(\varepsilon).
  \]
  Equation~\eqref{eq:Phizero} gives $\Phi_{\lambda_*}(q_*,r_*)=0$, that is,
  \[
    r_*(1-\log r_*)+J_{\lambda_*}(r_*)
    +\frac{r_*}{1-q_*}J_{\lambda_*}(q_*)=0.
  \]
  Subtracting yields
  \[
    \frac{\SO+\BO-\lambda_*\FB}{a}
    =
    -\frac{r_*(q_*)^N}{1-q_*}J_{\lambda_*}(q_*)
    +O(\varepsilon)
    =
    O\!\left(\varepsilon+(q_*)^N\right),
  \]
  with a constant independent of $N$ and $\varepsilon$.

  Finally, since $\rho_*=\lambda_*/2$ and $\RO=(\SO+\BO)/2$,
  \[
    \frac{\RO}{\FB}-\rho_*
    =
    \frac{(\SO+\BO-\lambda_*\FB)/a}{2\,\FB/a}.
  \]
  Combining the numerator estimate with $\FB/a\ge c>0$ gives
  \[
    \frac{\RO}{\FB}
    =
    \rho_*
    +O\!\left(\varepsilon+(q_*)^N\right).
  \]

  Hence, for any $\eta>0$, first choose a sufficiently large finite $N$
  so that $(q_*)^N$ is small enough, and then choose a sufficiently small
  positive $\varepsilon$ to obtain
  \[
    \left|\frac{\RO}{\FB}-\rho_*\right|<\eta.
  \]
  Therefore $\rho_{\mathrm{RO}}\le \rho_*<0.436944$.
\end{proof}

\section*{AI Disclosure}

Generative AI tools, including OpenAI Codex and GPT, played a substantive role
in the mathematical exploration underlying this work. Codex was first used to
search numerically for improved constructions and was subsequently prompted to
infer analytic functional forms from the solutions it found. From the resulting
candidate forms, the authors identified the recurring structural pattern
underlying the construction developed in this paper. GPT was then used to
assist in developing rigorous proofs of the resulting statements. With AI
assistance, the authors developed the construction and its analysis, selected
and reorganized the essential arguments into the proof structure presented
here, and improved their mathematical interpretation and exposition. The
authors checked the analytic derivations and proofs and take full
responsibility for the correctness and originality of the paper.

\begingroup
\raggedright
\bibliographystyle{unsrt}
\bibliography{references}

@article{MS83,
  author = {Roger B. Myerson and Mark A. Satterthwaite},
  title = {{Efficient Mechanisms for Bilateral Trading}},
  journal = {Journal of Economic Theory},
  year = {1983},
  volume = {29},
  number = {2},
  pages = {265--281},
  doi = {10.1016/0022-0531(83)90048-0},
  url = {https://doi.org/10.1016/0022-0531(83)90048-0}
}

@inproceedings{BCWZ17,
  author = {Johannes Br\"ustle and Yang Cai and Fa Wu and Mingfei Zhao},
  title = {{Approximating Gains from Trade in Two-sided Markets via Simple Mechanisms}},
  booktitle = {Proceedings of the 18th ACM Conference on Economics and Computation},
  year = {2017},
  pages = {589--590},
  publisher = {Association for Computing Machinery},
  doi = {10.1145/3033274.3085148},
  eprint = {1706.04637},
  archivePrefix = {arXiv},
  url = {https://doi.org/10.1145/3033274.3085148}
}

@inproceedings{CGMZ21,
  author = {Yang Cai and Kira Goldner and Steven Ma and Mingfei Zhao},
  title = {{On Multi-Dimensional Gains from Trade Maximization}},
  booktitle = {Proceedings of the 2021 ACM-SIAM Symposium on Discrete Algorithms},
  year = {2021},
  pages = {1079--1098},
  publisher = {Society for Industrial and Applied Mathematics},
  doi = {10.1137/1.9781611976465.67},
  note = {Full version: \href{https://arxiv.org/abs/2007.13934v2}{arXiv:2007.13934v2}},
  url = {https://doi.org/10.1137/1.9781611976465.67}
}

@misc{BDK21,
  author = {Moshe Babaioff and Shahar Dobzinski and Ron Kupfer},
  title = {{A Note on the Gains from Trade of the Random-Offerer Mechanism}},
  year = {2021},
  howpublished = {\href{https://arxiv.org/abs/2111.07790}{arXiv:2111.07790}},
  eprint = {2111.07790},
  archivePrefix = {arXiv},
  url = {https://arxiv.org/abs/2111.07790}
}

@misc{CGLM26,
  author = {Yang Cai and Vineet Gupta and Zun Li and Aranyak Mehta},
  title = {{A New Lower Bound for the Random Offerer Mechanism in Bilateral Trade using AI-Guided Evolutionary Search}},
  year = {2026},
  howpublished = {\href{https://arxiv.org/abs/2603.08679}{arXiv:2603.08679}},
  eprint = {2603.08679},
  archivePrefix = {arXiv},
  url = {https://arxiv.org/abs/2603.08679}
}

@misc{Jo26,
  author = {Sunghyeon Jo},
  title = {{Tighter Bounds for the Random-Offerer Mechanism in Bilateral Trade}},
  year = {2026},
  howpublished = {\href{https://arxiv.org/abs/2607.13959}{arXiv:2607.13959}},
  eprint = {2607.13959},
  archivePrefix = {arXiv},
  url = {https://arxiv.org/abs/2607.13959}
}

@inproceedings{DMSW22,
  author = {Yuan Deng and Jieming Mao and Balasubramanian Sivan and Kangning Wang},
  title = {{Approximately Efficient Bilateral Trade}},
  booktitle = {Proceedings of the 54th Annual ACM SIGACT Symposium on Theory of Computing},
  year = {2022},
  pages = {718--721},
  publisher = {Association for Computing Machinery},
  doi = {10.1145/3519935.3520054},
  url = {https://doi.org/10.1145/3519935.3520054}
}

@inproceedings{Fei22,
  author = {Yumou Fei},
  title = {{Improved Approximation to First-Best Gains-from-Trade}},
  booktitle = {Web and Internet Economics (WINE 2022)},
  editor = {Kristoffer Arnsfelt Hansen and Tracy Xiao Liu and Azarakhsh Malekian},
  series = {Lecture Notes in Computer Science},
  volume = {13778},
  year = {2022},
  pages = {204--218},
  publisher = {Springer},
  doi = {10.1007/978-3-031-22832-2_12},
  eprint = {2205.00140},
  archivePrefix = {arXiv},
  url = {https://doi.org/10.1007/978-3-031-22832-2_12}
}

@misc{HW25,
  author = {Jason Hartline and Kangning Wang},
  title = {{A Geometric Analysis of Gains from Trade}},
  year = {2025},
  howpublished = {\href{https://arxiv.org/abs/2508.06469}{arXiv:2508.06469}},
  eprint = {2508.06469},
  archivePrefix = {arXiv},
  url = {https://arxiv.org/abs/2508.06469}
}

@misc{BWZ26,
  author = {Xiaohui Bei and Wenhao Wu and Shengwei Zhou},
  title = {{The Power of Simple Mechanisms: A Tight \(e\)-Approximation for Gains from Trade in Matching Markets}},
  year = {2026},
  howpublished = {\href{https://arxiv.org/abs/2609.37997}{arXiv:2609.37997}},
  eprint = {2609.37997},
  archivePrefix = {arXiv},
  url = {https://arxiv.org/abs/2609.37997}
}
\endgroup

\appendix
\section{Proofs}\label{sec:proofs}

To prove the lemmas, we first analyze optimal offers under the hard
distributions and the properties of $C,D,h,H$.

In the proofs of the offer formulas and the recurrence below,
\(F\in\mathcal F\) denotes an arbitrary inner distribution, and
\(G=T_{x,\varepsilon}F\in\mathcal F\) denotes the distribution obtained by
adding a layer. We use the parameter conditions of Lemma~\ref{prop:recursion}.
The functions \(h_F,H_F\) and \(h_G,H_G\) correspond to these two seller
distributions, respectively. Gains from trade and buyer values are always
paired with the seller distribution currently under consideration.

\subsection{Optimal offers and surplus structure under the transformation}
\label{subsec:formulas}

We first establish a basic invariant of \(\mathcal F\) and use it to
characterize optimal offers for high buyer values.
We then explicitly compute the optimal offers and the corresponding
\(C_G,D_G,h_G\) after one transformation \(G=T_{x,\varepsilon}F\).

\begin{lemma}\label{lem:invariant-boundary}
  For every \(F\in\mathcal F\) with upper support endpoint \(z\),
  \begin{equation}
    (1-s)F(s)\le1-z\qquad(0\le s\le z).
    \label{eq:invariant}
  \end{equation}
  When \(b>1\), the optimal offer is \(z\), and
  \begin{equation}
    C_F(b)=D_F(b),\qquad h_F(b)=(1-\lambda)D_F(b)\ge0.
    \label{eq:fulltrade}
  \end{equation}
  In particular, \(H_F(u)=H_F(1)\) for \(u\ge1\).
  If \(a=a(F)\) is the lower support endpoint, then \(H_F(u)=0\) for
  \(0\le u<a\), as stated in~\eqref{eq:Htail}.
\end{lemma}

\begin{proof}
  Equation~\eqref{eq:invariant} holds for \(F=\delta_1\).
  Suppose it holds for \(F\), and let \(G=T_{x,\varepsilon}F\),
  \(\ell=\varepsilon z\), and \(z'=1-x(1-\ell)\).
  For \(0\le s=\varepsilon t\le\ell\),
  \begin{equation*}
    \begin{aligned}
      (1-s)G(s)
       & =x\bigl[\varepsilon(1-t)F(t)+(1-\varepsilon)F(t)\bigr] \\
       & \le x\bigl[\varepsilon(1-z)+1-\varepsilon\bigr]
      =x(1-\ell)=1-z'.
    \end{aligned}
  \end{equation*}
  Equality holds on \([\ell,z']\), so~\eqref{eq:invariant} follows by induction.

  When \(b>1\) and \(0\le p<z\), the inequality \(F(p)<1\) gives
  \[
    (b-p)F(p)=(1-p)F(p)+(b-1)F(p)<1-z+(b-1)=b-z.
  \]
  When \(z<p\le b\), the offer yields profit \(b-p<b-z\).
  Thus, the unique optimal offer is \(z\), all seller types accept, and
  \eqref{eq:fulltrade} follows.
  Since \(H_F(1)\le h_F(0)=0\), the nonnegative surplus gaps for \(b>1\)
  cannot lower the minimum. Hence \(H_F(u)=H_F(1)\) for \(u\ge1\).

  When \(0\le u<a\), every \(0\le b\le u\) admits only offers at most
  \(b<a\), so no trade occurs.
  Thus \(C_F(b)=D_F(b)=h_F(b)=0\), and taking the minimum gives \(H_F(u)=0\).
\end{proof}

\begin{proposition}\label{prop:transformed-price}
  Let \(G=T_{x,\varepsilon}F\), where \(F\in\mathcal F\) has upper support
  endpoint \(z\), and write \(\ell=\varepsilon z\) and \(z'=1-x(1-\ell)\).
  Then
  \begin{equation}
    p_G(b)=
    \begin{cases}
      \varepsilon p_F(b/\varepsilon), & 0\le b\le\varepsilon, \\
      \ell,                           & \varepsilon<b\le1,    \\
      z',                             & b>1.
    \end{cases}
    \label{eq:pG}
  \end{equation}
\end{proposition}

\begin{proof}
  Profit decreases with the price above the upper support endpoint, so it
  suffices to consider \(0\le p\le\min\{b,z'\}\).

  When \(0\le b\le\varepsilon\), write \(p=\varepsilon t\) on the inner
  interval \([0,\ell]\). Then
  \[
    (b-p)G(p)=x\varepsilon(b/\varepsilon-t)F(t).
  \]
  The lowest optimal offer on this interval is therefore
  \(\varepsilon p_F(b/\varepsilon)\).
  If \(b<\ell\), there is no feasible offer in the outer layer.
  If \(b\ge\ell\), then \(b\le\varepsilon<1\) gives
  \[
    \frac{d}{dp}\left(\frac{x(1-\ell)(b-p)}{1-p}\right)
    =\frac{x(1-\ell)(b-1)}{(1-p)^2}<0.
  \]
  Hence, no feasible outer-layer offer is better than \(\ell\), which is
  already included in the inner comparison.
  This proves the first case of~\eqref{eq:pG}.

  When \(\varepsilon<b\le1\), we have \(b/\varepsilon>1\).
  Lemma~\ref{lem:invariant-boundary} gives \(z\) as the unique optimal offer
  in the old distribution, so \(\ell\) is the unique optimal inner offer.
  For feasible offers on the outer interval \([\ell,z']\),
  \[
    \frac{d}{dp}\left(\frac{x(1-\ell)(b-p)}{1-p}\right)
    =\frac{x(1-\ell)(b-1)}{(1-p)^2}\le0.
  \]
  Thus, \(\ell\) is also optimal on this interval.
  When \(b=1\), the profit above is constant, and the convention of choosing
  the lowest optimal offer again selects \(\ell\).

  When \(b>1\), the distribution \(G\in\mathcal F\) has upper support
  endpoint \(z'\), so Lemma~\ref{lem:invariant-boundary} gives \(p_G(b)=z'\).
\end{proof}

\begin{lemma}\label{lem:transformed-gains}
  Under the conditions of Proposition~\ref{prop:transformed-price}, assume
  further that \(0<\varepsilon<1-x\). Then
  \begin{equation}
    C_G(b)=
    \begin{cases}
      x\varepsilon C_F(b/\varepsilon), & 0\le b\le\varepsilon, \\
      x\varepsilon D_F(z)+x(b-\ell),   & \varepsilon<b\le1,    \\
      D_G(b),                          & b>1.
    \end{cases}
    \label{eq:CG}
  \end{equation}
  Let \(k(b)=(b-\ell)/(1-\ell)\). We have
  \begin{equation}
    D_G(b)=
    \begin{cases}
      x\varepsilon D_F(b/\varepsilon),       & 0\le b\le\ell, \\
      x\varepsilon D_F(z)+(1-\ell)D_x(k(b)), & b>\ell.
    \end{cases}
    \label{eq:DG}
  \end{equation}
  The first case of \(C_G\) ends at \(\varepsilon\), whereas the first case
  of \(D_G\) ends at \(\ell\le\varepsilon\).
  To describe the discrepancy uniformly on \([0,\varepsilon]\), define
  \begin{equation}
    \Delta(b)=
    \begin{cases}
      0,                                               & 0\le b\le\ell,        \\
      x\displaystyle\int_\ell^b\frac{s-\ell}{1-s}\,ds, & \ell<b\le\varepsilon.
    \end{cases}
    \label{eq:Delta}
  \end{equation}
  Then \(D_G(b)=x\varepsilon D_F(b/\varepsilon)+\Delta(b)\) for
  \(0\le b\le\varepsilon\), and
  \begin{equation}
    h_G(b)=
    \begin{cases}
      x\varepsilon h_F(b/\varepsilon)-\lambda\Delta(b),          & 0\le b\le\varepsilon, \\
      x\varepsilon(1-\lambda)D_F(z)+(1-\ell)h_{\lambda,x}(k(b)), & \varepsilon<b\le1,    \\
      (1-\lambda)D_G(b),                                         & b>1.
    \end{cases}
    \label{eq:hG}
  \end{equation}
\end{lemma}

\begin{proof}
  Substituting the optimal offers from Proposition~\ref{prop:transformed-price}
  into the definition of \(C_G\) gives~\eqref{eq:CG}.

  For \(b\le\ell\), substituting \(s=\varepsilon t\) in
  \(D_G(b)=\int_0^bG(s)\,ds\) gives the first case of~\eqref{eq:DG}.
  For \(s\ge\ell\), we have \(G(s)=Q_x(k(s))\).
  Using \(k(s)\) as the integration variable normalizes the new hyperbolic
  layer to the standard layer \(Q_x\), giving
  \[
    D_G(b)=x\varepsilon D_F(z)+(1-\ell)\int_0^{k(b)}Q_x(t)\,dt,
    \qquad b>\ell,
  \]
  which is the second case of~\eqref{eq:DG}.

  When \(b\le\ell\), we have \(\Delta(b)=0\).
  When \(\ell<b\le\varepsilon\), we have \(F(b/\varepsilon)=1\) and
  \(\varepsilon<1-x<z'\), so
  \[
    D_G(b)-x\varepsilon D_F(b/\varepsilon)
    =\int_\ell^b\bigl[G(s)-x\bigr]\,ds
    =x\int_\ell^b\frac{s-\ell}{1-s}\,ds=\Delta(b).
  \]
  Combining this identity with~\eqref{eq:CG} and using \(h=C-\lambda D\)
  gives~\eqref{eq:hG}. The second case also uses
  \(x(b-\ell)=(1-\ell)C_x(k(b))\).
\end{proof}

\subsection{Proof of Lemma~\ref{prop:recursion}}\label{subsec:lemmaproofs}
We first establish three auxiliary lemmas.

\begin{lemma}\label{lem:first}
  Uniformly for \(0\le u\le\varepsilon\),
  \begin{equation*}
    H_G(u)=x\varepsilon H_F(u/\varepsilon)+O(u^2).
  \end{equation*}
\end{lemma}

\begin{proof}[Proof of Lemma~\ref{lem:first}]
  Take \(0<\varepsilon<\min\{1/2,1-x\}\).
  By~\eqref{eq:Delta}, for \(\ell<b\le\varepsilon\),
  \[
    0\le\Delta(b)
    \le x\int_\ell^b 2s\,ds
    \le xb^2.
  \]
  For \(0\le b\le\ell\), we have \(\Delta(b)=0\).
  Thus, \(\Delta(b)=O(b^2)\) holds uniformly on \([0,\varepsilon]\), with
  a constant independent of \(F\).
  For any \(0\le b\le u\le\varepsilon\), equation~\eqref{eq:hG} therefore gives
  \[
    h_G(b)=x\varepsilon h_F(b/\varepsilon)+O(u^2).
  \]
  Taking the minimum over \(b\in[0,u]\) on both sides proves the claim,
  with an implicit constant independent of \(F\).
\end{proof}

\begin{lemma}\label{lem:second}
  For \(\varepsilon\le u\le\tau\),
  \begin{equation*}
    H_G(u)=H_G(\varepsilon).
  \end{equation*}
\end{lemma}

\begin{proof}[Proof of Lemma~\ref{lem:second}]
  When \(\varepsilon<b\le\tau\), we have \(0\le k(b)\le b\le\tau\).
  Equations~\eqref{eq:idealh} and~\eqref{eq:hG} give
  \[
    h_G(b)=x\varepsilon(1-\lambda)D_F(z)+(1-\ell)h_{\lambda,x}(k(b))\ge0.
  \]
  Since \(H_G(\varepsilon)\le h_G(0)=0\),
  \[
    H_G(u)=H_G(\varepsilon)\qquad(\varepsilon\le u\le\tau).\qedhere
  \]
\end{proof}

\begin{lemma}\label{lem:third}
  For \(u\ge\tau\),
  \begin{equation*}
    H_G(u)=H_{\lambda,x}(u)+O(\varepsilon).
  \end{equation*}
\end{lemma}

\begin{proof}[Proof of Lemma~\ref{lem:third}]
  Take \(0<\varepsilon<\min\{1/2,\tau/2\}\).
  Since \(\tau>1-\lambda\) and \(x<1-\tau\), we have \(x<\lambda\).
  By~\eqref{eq:idealh}, on \([0,1]\),
  \[
    h'_{\lambda,x}(b)=
    \begin{cases}
      x-\dfrac{\lambda x}{1-b}, & 0\le b\le1-x, \\
      x-\lambda,                & 1-x\le b\le1.
    \end{cases}
  \]
  When \(0\le b\le1-x\), we have \(\lambda x/(1-b)\le\lambda\), so
  \[
    x-\lambda\le h'_{\lambda,x}(b)\le x(1-\lambda).
  \]
  Together with the derivative on the second interval and \(x<\lambda\),
  this gives \(|h'_{\lambda,x}(b)|\le\lambda\).
  Therefore, for \(s,t\in[0,1]\),
  \[
    |h_{\lambda,x}(s)-h_{\lambda,x}(t)|\le\lambda|s-t|,
    \qquad |h_{\lambda,x}(b)|\le\lambda b,
  \]
  where the second inequality uses \(h_{\lambda,x}(0)=0\).

  For \(\varepsilon<b\le1\), we have \(0\le k(b)\le b\le1\) and
  \((1-\ell)|k(b)-b|=\ell(1-b)\).
  Equation~\eqref{eq:hG} gives
  \begin{equation*}
    \begin{aligned}
      |h_G(b)-h_{\lambda,x}(b)|
       & \le x\varepsilon(1-\lambda)D_F(z)
      +(1-\ell)|h_{\lambda,x}(k(b))-h_{\lambda,x}(b)|
      +\ell|h_{\lambda,x}(b)|                                             \\
       & \le x\varepsilon(1-\lambda)D_F(z)+\lambda\ell(1-b)+\lambda\ell b \\
       & \le \bigl[x(1-\lambda)+\lambda\bigr]\ell
      \le\ell\le\varepsilon,
    \end{aligned}
  \end{equation*}
  Here we used \(D_F(z)\le z\), \(\ell=\varepsilon z\), and \(x\le1\).

  Let \(M=H_G(\varepsilon)\).
  Since \(h_G(b)\ge-\lambda D_G(b)\ge-\lambda b\) and \(h_G(0)=0\),
  we have \(-\lambda\varepsilon\le M\le0\).
  Lemma~\ref{lem:second} gives \(H_G(u)=M\) for \(\varepsilon\le u\le\tau\).
  Hence \(h_G(t)\ge M\ge-\varepsilon\) for every \(0\le t\le\tau\).

  Fix \(\tau\le u\le1\). Since \(H_G(u)\le h_G(u)\),
  \[
    H_G(u)\le h_{\lambda,x}(u)+\varepsilon.
  \]
  For \(0\le t\le\tau\), the inequality \(h_{\lambda,x}(u)\le0\) gives
  \[
    h_G(t)\ge-\varepsilon\ge h_{\lambda,x}(u)-\varepsilon.
  \]
  For \(\tau\le t\le u\), since \(h_{\lambda,x}\) is decreasing on \([\tau,1]\),
  \[
    h_G(t)\ge h_{\lambda,x}(t)-\varepsilon
    \ge h_{\lambda,x}(u)-\varepsilon.
  \]
  Taking the minimum over \(t\in[0,u]\) gives
  \(H_G(u)\ge h_{\lambda,x}(u)-\varepsilon\).
  By~\eqref{eq:idealH},
  \[
    |H_G(u)-H_{\lambda,x}(u)|
    =|H_G(u)-h_{\lambda,x}(u)|\le\varepsilon
    \qquad(\tau\le u\le1).
  \]
  When \(u\ge1\), the two functions equal \(H_G(1)\) and
  \(H_{\lambda,x}(1)\), respectively, so the same error bound holds.
  The bound is independent of \(F\).
\end{proof}

The three lemmas control the errors on the respective intervals.
Combining them proves Lemma~\ref{prop:recursion}.

\begin{proof}[Proof of Lemma~\ref{prop:recursion}]
  Set \(\varepsilon_0=\min\{1/2,\tau/2\}\), and assume
  \(0<\varepsilon<\varepsilon_0\).
  The proofs of the three lemmas above apply throughout this range, and all
  error constants depend only on \(\lambda,x\).
  The first interval uses \(\Delta(b)\le xb^2\), and the third uses
  \(\ell\le\varepsilon\),
  \(|k(b)-b|\le\varepsilon/(1-\varepsilon)\le2\varepsilon\), and
  \(0\le D_F(z)\le z\le1\).
  These bounds and \(\varepsilon_0\) are independent of \(F\) and its
  recursion depth.
  Also, \(C_F(b)\ge0\), \(D_F(b)\le b\), and \(h_F(0)=0\) give
  \[
    -\lambda\le H_F(1)\le0.
  \]
  By Lemma~\ref{lem:first}, substituting \(t=u/\varepsilon\) yields
  \begin{equation*}
    \begin{aligned}
      \int_0^\varepsilon\frac{H_G(u)}{u^2}\,du
       & =x\int_0^1\frac{H_F(t)}{t^2}\,dt+O(\varepsilon)   \\
       & =x\bigl[L_\lambda(F)-H_F(1)\bigr]+O(\varepsilon),
    \end{aligned}
  \end{equation*}
  where we used~\eqref{eq:Htail}.
  Lemma~\ref{lem:second} and the estimate in Lemma~\ref{lem:first} at
  \(u=\varepsilon\) give
  \begin{equation*}
    \begin{aligned}
      \int_\varepsilon^\tau\frac{H_G(u)}{u^2}\,du
       & =\bigl[x\varepsilon H_F(1)+O(\varepsilon^2)\bigr]
      \left(\frac1\varepsilon-\frac1\tau\right)            \\
       & =xH_F(1)+O(\varepsilon).
    \end{aligned}
  \end{equation*}
  This uses \(|H_F(1)|\le\lambda\), so the error bound remains independent
  of \(F\).
  By Lemma~\ref{lem:third} and~\eqref{eq:idealH},
  \[
    \int_\tau^\infty\frac{H_G(u)}{u^2}\,du
    =J_\lambda(x)+O(\varepsilon),
  \]
  Adding the integrals over the three intervals gives
  \[
    L_\lambda(G)=xL_\lambda(F)+J_\lambda(x)+O(\varepsilon).
  \]
  More explicitly, bound their respective remainders by \(C_1\varepsilon\),
  \(C_2\varepsilon\), and \(C_3\varepsilon\), where \(C_1,C_2,C_3\) depend
  only on \(\lambda,x\). Then
  \[
    \left|L_\lambda(G)-xL_\lambda(F)-J_\lambda(x)\right|
    \le(C_1+C_2+C_3)\varepsilon.
  \]
  Taking \(C=C_1+C_2+C_3\) proves~\eqref{eq:recursion}, and the same
  \(C,\varepsilon_0\) work for all \(F\in\mathcal F\).
\end{proof}

\subsection{Proof of Lemma~\ref{lem:unroll}}\label{subsec:unrollproof}

\begin{proof}[Proof of Lemma~\ref{lem:unroll}]
  For \(\delta_1\), we have \(H_{\lambda,\delta_1}\equiv0\), hence
  \(L_\lambda(\delta_1)=0\).
  Write the constants supplied by Lemma~\ref{prop:recursion} at \(x=q,r\)
  as \(C_q,C_r\) and \(\varepsilon_{0,q},\varepsilon_{0,r}\), respectively.
  Set \(\varepsilon_0=\min\{\varepsilon_{0,q},\varepsilon_{0,r}\}\) and take
  \(0<\varepsilon<\varepsilon_0\).
  Let \(L_k=L_\lambda(X_{q,\varepsilon,k})\). By~\eqref{eq:recursion},
  \[
    L_0=0,\qquad
    L_k=qL_{k-1}+J_\lambda(q)+e_k,
    \qquad |e_k|\le C_q\varepsilon\quad(1\le k\le N).
  \]
  Successive substitution gives
  \[
    L_N=\sum_{k=1}^N q^{N-k}\bigl[J_\lambda(q)+e_k\bigr]
    =\frac{1-q^N}{1-q}J_\lambda(q)+\sum_{k=1}^N q^{N-k}e_k.
  \]
  Applying~\eqref{eq:recursion} once more to the outermost layer yields
  \[
    L_\lambda(G_{q,r,N,\varepsilon})=rL_N+J_\lambda(r)+e_{\mathrm{out}},
    \qquad |e_{\mathrm{out}}|\le C_r\varepsilon.
  \]
  The total remainder therefore satisfies
  \[
    \left|r\sum_{k=1}^N q^{N-k}e_k+e_{\mathrm{out}}\right|
    \le\left(rC_q\frac{1-q^N}{1-q}+C_r\right)\varepsilon
    \le\left(\frac{rC_q}{1-q}+C_r\right)\varepsilon.
  \]
  When \(N=0\), all empty sums are \(0\).
  Since \(0<q<1\), both the error bound and the restriction on \(\varepsilon\)
  are independent of \(N\), proving the claim.
\end{proof}

\subsection{Proof of Lemma~\ref{lem:bo}}\label{subsec:SOproof}

We first analyze seller offers for an arbitrary \(F\in\mathcal F\) and
the outer distribution \(G=T_{r,\varepsilon}F\), where \(0<r<1-\tau\).
Take \(B=B_{\lambda,G}\), and write \(U=U_G\) and \(a=a(G)>0\).
We then apply the result to \(F=X_{q,\varepsilon,N}\) and
\(G=G_{q,r,N,\varepsilon}\in\mathcal F\).
We have \(h_{\lambda,r}(1)<0\), and the proof of Lemma~\ref{lem:third} gives
\(h_G(1)=h_{\lambda,r}(1)+O(\varepsilon)\).
On the other hand, \(h_G(b)\ge-\lambda b\) gives
\[
  \min_{0\le b\le\varepsilon}h_G(b)=O(\varepsilon).
\]
Thus, for sufficiently small \(\varepsilon\),
\[
  h_G(1)<\min_{0\le b\le\varepsilon}h_G(b),
\]
The required restriction on the scale depends only on \(\lambda,r\),
independently of \(F\).
On \((\varepsilon,1]\), equation~\eqref{eq:hG} gives
\[
  h_G(b)=r\varepsilon(1-\lambda)D_F(z)
  +(1-\ell)h_{\lambda,r}(k(b)).
\]
The unique minimizer of \(h_{\lambda,r}\) on \([0,1]\) is \(1\), and
\(k(1)=1\), so \(h_G(b)>h_G(1)\) for \(\varepsilon<b<1\).
For \(b>1\), we have \(h_G(b)\ge0\).
Thus, the unique global minimizer of \(h_G\) is \(1\).
Equation~\eqref{eq:buyer} gives
\begin{equation}
  m_{\lambda,G}(u)=1\quad(u\ge1),\qquad B\le1,\qquad \P(B=1)=\P(U\ge1)=a.
  \label{eq:topatom}
\end{equation}

For a seller with cost \(s\in[a,z']\), offering price \(1\) yields profit
\(a(1-s)>0\).
For any \(s\le p<1\), the inequality \(B\le U\) gives
\[
  (p-s)\P(B\ge p)\le(p-s)\frac ap
  =a\left(1-\frac sp\right)<a(1-s).
\]
Offers below \(s\) yield nonpositive profit, and offers above \(1\) are
never accepted. Thus, every seller's unique optimal offer is \(1\), and
\begin{equation}
  \SO=a\E[1-S]=aD_G(1).
  \label{eq:SOexact}
\end{equation}

By~\eqref{eq:DG}, noting that \(k(1)=(1-\ell)/(1-\ell)=1\), we have
\begin{equation}
  \begin{aligned}
    D_G(1)
     & =r\varepsilon D_F(z)+(1-\ell)D_r(k(1)) \\
     & =O(\varepsilon)+(1-\ell)D_r(1)         \\
     & =D_r(1)+O(\varepsilon)                 \\
     & =-r\log r+r+O(\varepsilon).
  \end{aligned}
  \label{eq:Dbound}
\end{equation}
Here \(0\le D_F(z)\le z\le1\) gives
\(0\le r\varepsilon D_F(z)\le r\varepsilon\), while
\(\ell=\varepsilon z\le\varepsilon\) gives
\(0\le\ell D_r(1)\le\varepsilon D_r(1)\).
Thus, both remainders above are \(O(\varepsilon)\), with implicit constants
depending only on \(r\).
Taking \(F=X_{q,\varepsilon,N}\) yields~\eqref{eq:SOannounced} for the
specific construction.
The scale restriction depends only on \(\lambda,r\), and the error bound
is independent of the inner distribution \(F\), \(q\), and \(N\).

\end{document}